\documentclass[12pt]{article}
\usepackage{amssymb,amsmath,amsthm,a4wide}
\usepackage{url,graphics,graphicx}

\advance\textheight by 1.8cm
\DeclareMathOperator{\C}{C}
\DeclareMathOperator{\ceil}{ceil}
\newcommand{\alg}[1]{\hspace{1pt}\operatorname{\mathtt{#1}}\hspace{1pt}}
\newcommand{\varalg}[1]{\operatorname{\scalebox{.92}{$\mathrm{#1}$}}}

\renewcommand{\mid}{\mathop{|}}

\newtheorem{definition}{Definition}
\newtheorem{theorem}{Theorem}
\newtheorem{proposition}[theorem]{Proposition}
\newtheorem{lemma}[theorem]{Lemma}
\newtheorem{corollary}[theorem]{Corollary}
\newtheorem*{corollary*}{Corollary}
\begin{document}

\title{Optimal probabilistic polynomial time compression and the Slepian-Wolf theorem: 
tighter version and simple proofs}
\author{Bruno Bauwens\footnote{
National Research University Higher School of Economics (HSE),
Faculty of Computer Science,
Kochnovskiy Proezd 3, Moscow, 125319 Russia,
E-mail: \texttt{Brbauwens at gmail dot com}
  }
}

\date{}

\maketitle  

\vspace{-.8cm}

\begin{abstract}
  \noindent
  We give simplify the proofs of the 2 results in Marius Zimand's paper 
  {\em Kolmogorov complexity version of Slepian-Wolf coding, proceedings
  of STOC 2017, p22--32.} 
  The first is a universal polynomial time compression algorithm: 
  on input $\varepsilon > 0$, a number $k$ and a string~$x$, computes in 
  polynomial time with probability $1-\varepsilon$ a program of length
  $k + O(\log^2 (|x|/\varepsilon))$ that outputs~$x$, provided that there 
  exists such a program of length at most~$k$.
  The second result, is a universal distributed compression algorithm,
  in which several parties each send some string to a common receiver. 
  Marius Zimand proved a variant of the Slepian-Wolf theorem 
  using Kolmogorov complexity (in stead of Shannon entropy). 
  With our simpler proof we improve the parameters of Zimand's result.
\end{abstract}

\section{Generating almost shortest programs in polynomial time.}

Assume there exists a program of length at most $k$ that (without input) produces $x$.
Given~$x$, can we find such a program quickly? A trivial way would be to run
all programs of length at most $k$ in parallel and output the first one that
produces~$x$.  Compressing $x$ in this way, requires an
exponentially longer time than the fastest decompression program.
Can we compress faster? Can we compress even faster than the smallest decompression time? 

It was known that many strings can be compressed rapidly using pseudo random
generators and we explain the technique in the next paragraph. Under some
computational complexity assumption this implies that strings exist who can be
compressed much faster than the time required by any decompression algorithm.

In \cite{Zimand2017STOC}, improving results 
from~\cite{BauwensShortlistJournal,BauwensLinearList,teutschShort}, 
Marius Zimand gave a probabilistic algorithm that compresses {\em every} string with
high probability to an almost shortest description.  
This implies the unconditional existence of strings that can be compressed in
polynomial time by a randomized algorithm, but can not be decompressed in for
example double exponential time, or even in a time bounded by the Busy Beaver
function of half of the length.


Let $\C(x|z)$ represent the Kolmogorov complexity of~$x$ conditional to~$z$. 
(See~\cite{LiVitanyiForthEdition,bookShenVereshchagin} for background on Kolmogorov complexity.)

\begin{theorem}[Zimand, STOC 2017]\label{th:fastCompression}
  There exist a deterministic algorithm $\alg{D}$ and a probabilistic algorithm 
  that on input $\varepsilon \in (0,1]$, $x$ and $k\ge C(x|z)$ produces with probability $1-\varepsilon$ 
  in time polynomial in $|x|/\varepsilon$, 
  a string $p$ of length $k + O(\log^2 (|x|/\varepsilon))$ such that ${\alg D}(p,z) =x$. 
\end{theorem}

\noindent
Note that the probabilistic compression algorithm does not use~$z$. 
Moreover, we can assume that this algorithm uses at most $O(\log^2 (|x|/\varepsilon))$ random bits.

\subsection{Compressing a string using pseudorandomness.}

Imagine Alice wants to send $x$ to Bob using a minimal number of bits.
If we assume that Alice has unlimited computational resources and even access to the halting 
problem, she can compute a shortest program for $x$ and communicate $C(x)$ many bits. 
Now assume Alice can only compute in probabilistic polynomial time,
and Bob's has any computational resources available.

Let $|x|=n$. If additionally Alice and Bob know~$C(x)$, and they have shared randomness, Alice can communicate $x$ using 
$C(x) + O\left(\log \tfrac{n}{\varepsilon}\right)$~bits. 
Let $r_1, \dots, r_{\ell}$ be a list of $\ell = \C(x) + \ceil\log \tfrac{2n}{\varepsilon}$ random strings of length~$n$,
Alice sends $n$ and the bits $b_i = x \cdot r = \sum_j x_j r_{i,j} \bmod 2$ 
for $i = 1, \dots, \ell$.
Clearly, Alice can do this in time polynomial in $n/\varepsilon$.

With probability $1-\varepsilon$, Bob recovers~$x$ as follows:
he enumerates all $n$-bit strings $y$ with $\C(y) \le \C(x)$.
Note that there are less than $2^{\C(x)+1}$ such strings. 
For each enumerated~$y$ he checks 
whether $y \cdot r_j = b_j$ for all $j$.
He outputs the first~$y$ for which this is true.

Why does Bob recover $x$ with probability~$1-\varepsilon$? 
If $x$ and $y$ are different $n$-bit strings, then for a random string $w$, 
the condition $x \cdot r = y \cdot r$ holds with probability~$1/2$. 
The probability that this true for all strings $r_j$, is at most~$2^{-\ell} \le \varepsilon 2^{-\C(x)-1}$.
The algorithm considers less than $2^{\C(x)+1}$ strings~$y$. 
Hence, by the union bound, the probability that these bits match for some $j$
is at most~$\varepsilon$.

\medskip
To prove Theorem~\ref{th:fastCompression}, the main challenge is to get rid of
the public randomness.  People familiar with communication complexity may
suggest to apply the Newman trick to convert public randomness to private
randomness.  Unfortunately, this would make Alice's algorithm non-computable.

Another technique would be to apply the strategy above using pseudo random strings.
More precisely, Alice generates a random seed of logarithmic size, and uses a pseudorandomness 
generator to compute strings $r_j$.
Then she sends the parity bits and the random seed to Alice. 
Bob can now reconstructed the $r_j$, and hence also the string~$x$.
However, the strings $r_j$ are not random but pseudo random and they can only be used
with computations that are weak enough to be fooled by these strings.
Remarkably, Theorem~\ref{th:fastCompression} holds for all strings and even
without using any assumption from computational complexity.

\subsection{Overview of the proof}

The proof has two parts. 
First we show that extractor functions 
map strings $x$ of small Kolmogorov complexity to a lists such that most strings 
$y$ in the list can almost reconstruct~$x$, in other words, $\C(x|y)$ is very small.
In the second part, we use hashing to refine this mapping into a ``bijective'' one.
We start with the second part.

\medskip
For any probabilistic algorithm $\alg A$, let $\varalg{A}_{u,v,\dots,z}$ denote the random variable 
that represents the outcome of the algorithm when it is run on input $(u,v,\dots,z)$.
In the two parts described above we obtain the following results. 

\newcommand{\showTheProposition}{
  There exists a probabilistic polynomial time algorithm $\alg F$ that on input $\varepsilon > 0, k, x$ 
   outputs a $k$-bit string such that with probability $1-\varepsilon$:
  \[
  \C(x|{\varalg F}_{\varepsilon,k,x},z) \;\le\; \max \left( 0, C(x|z) - k \right) + {O}( \log^{2} \tfrac{|x|}{\varepsilon} ).
  \]
  }

\begin{proposition}\label{prop:extractorSmallComplexity}
  \showTheProposition
\end{proposition}

\begin{proposition}\label{prop:hashing}
  There exists a deterministic algorithm $\alg D$ and a probabilistic algorithm $\alg G$
  that on input $\varepsilon > 0$, $b$, and $x$ computes in time polynomial in $|x|/\varepsilon$ 
  a string of length $O(b + \log (|x|/\varepsilon))$
  such that for all $w$ for which $C(x|w) < b$:
  \[
  \Pr \left\{ \alg D(w,\varalg{G}_{\varepsilon,b,x},z) = x \right\}
  \;\;\ge\;\; (1-\varepsilon) 
  \]
\end{proposition}

\begin{proof}[Proof of  Theorem~\ref{th:fastCompression}.]
  Let $b(\varepsilon,x) = c+c\log \tfrac{|x|}{\varepsilon}$ where $c$ is large enough 
  such that $b$ exceeds the logarithmic term in Proposition~\ref{prop:extractorSmallComplexity}.
  $\varalg{E}_{\varepsilon,k,x}$ is given by 
  an efficient coding of the pair 
  $\left(\varalg{G}_{\varepsilon,b(\varepsilon,x),x}, \varalg{F}_{\varepsilon,k,x}\right)$, 
  where $\alg{F}$ and $\alg G$ are obtained from Propositions~\ref{prop:extractorSmallComplexity} and~\ref{prop:hashing}. 
  We conclude that on input $\varalg{E}_{\varepsilon,k,x}$ and $z$, algorithm~$\alg D$ outputs~$x$  
  with probability~$(1-\varepsilon)^2$.  
  Theorem~\ref{th:fastCompression} follows after rescaling~$\varepsilon$.
\end{proof}

\subsection{Proof of  Proposition~\ref{prop:hashing}}

{\em Algorithm $\alg G$} outputs a descriptions of the following 4 numbers: 
$|x|$, $b$, a random prime $p$ with bit length at most 
$2b+2\log\left(|x|/\varepsilon\right)+\tilde{c}$ 
(the constant $\tilde{c}$ will be determined later), 
and the value $x \bmod p$, where $x$ is interpreted as a number in binary.
These numbers are encoded using space~$O(b + \log (|x|/\varepsilon))$. 

\medskip
\noindent
{\em Algorithm $\alg D$} 
enumerates all $u$ of length $|x|$ such that $\C(u \mid w,z) < b$.
If some $u$ is found such that $u = x \bmod p$, then $\alg D$ outputs this~$u$ and halts, 
otherwise, $\alg D$ runs for ever.  

\medskip
Note that there are less than $2^b$ strings $u$ with $\C(u \mid \dots) < b$. 
To prove the inequality of the proposition, apply 
the following for $n = |x|$ and $s = 2^{b}$. 

\begin{quote}
 If $x, u_1, \dots, u_s$ are different natural numbers less than $2^n$ and $P$ 
 is a set of at least $sn/\varepsilon$ prime numbers, 
 then for a fraction $1-\varepsilon$ of primes $p$ in $P$:
 $x \bmod p \;\not\in\; \{u_1 \bmod p, \dots, u_s \bmod p\}$.
\end{quote}

\noindent
This claim is true because $x - u_1$ has at most $n$ different prime factors $p \ge 2$. 
Thus, $x-u_1 = 0 \bmod p$ for at most a fraction $\varepsilon/s$ of $p \in P$.
The same is true for $u_2, \dots, u_s$, thus the total fraction of bad primes is at most~$\varepsilon$.

\medskip
By the prime number theorem, the $t$th prime is at most $O(t \log t)$ and hence, 
its bitwise representation is at most $2\log t$ for large~$t$.
Thus, for some large $\tilde{c}$, the set $P$ of primes of bit length
$2b + 2\log (|x|/\varepsilon) + \tilde{c}$, contains at least $2^{b}|x|/\varepsilon$ primes.
The claim above implies that with probability $1-\varepsilon$, the value $x \bmod p$
will be unique in the set of all $u$ with small enough conditional complexity,
and hence, $\alg D$ outputs~$x$ with probability~$1-\varepsilon$.

\subsection{Proof of  Proposition~\ref{prop:extractorSmallComplexity}}

For a positive integer $n$, let $[n] = \{1,2,\dots, n\}$.

\begin{definition}\label{def:extractor}
A function $f\colon [N] \times [D] \rightarrow [M]$ is a {\em $(k,\varepsilon)$-extractor} if 
for all $A \subseteq [N]$ of size at least $2^k$ and for all $B \subseteq [M]$:
\[
\left| \Pr\left\{ f(X_A,X_D) \in B \right\} - \frac{|B|}{M} \right| \le \varepsilon,
\]
where $X_A$ and $X_{D}$ represent random values in $A$ and $[D]$.
\end{definition}

In the proof of  Proposition~\ref{prop:extractorSmallComplexity}, 
$\alg F$ is the algorithm that on input $x$ returns $f(x,i)$ for a random $i \in [D]$ and 
for a suitable extractor~$f$. 
The following proposition shows that with
high probability the complexity of $x$ is small relative to $\alg F$'s output.

\begin{lemma}\label{lem:extractorSmallComplexity}
  Let $\tilde{z} = (z,\varepsilon,f)$.
  For all $x$ such that $f\colon [N] \times [D] \rightarrow [M]$, is a $(\C(x|\tilde{z})-c,\varepsilon)$-extractor
  and for at least an $(1-2\varepsilon)$-fraction of $i$ in~$[D]$:
  \[
  \C(x \mathop{|} f(x,i),\tilde{z}) \,\;\le\;\, \C(x {\mid} \tilde{z}) - \log M + O(\log (D|x|/\varepsilon)).
  \]
\end{lemma}

\noindent
Later we choose $f$ from a computable parameterized family so that 
it can be described by at most $O(\log (|x|/\varepsilon))$ bits. 
Hence, $f$'s presence in the condition of the complexity terms 
change them only by additive logarithmic terms. Without loss of generality, 
we can assume that $\varepsilon$ in  Theorem~\ref{th:fastCompression} 
satisfies $\C(\varepsilon) \le O(\log (1/\varepsilon)) \le O(\log (|x|/\varepsilon))$ as well.
In the proof of the proposition we use a slight variant of a property of extractor graphs 
that Marius Zimand calls ``rich owner property''. Here it is defined in a complementary way.

\begin{definition}\label{def:poorElements}
  Let $f\colon [N] \times [D] \rightarrow [M]$ and $S \subseteq [N]$. 
  \begin{itemize}
    \item 
      $y \in [M]$ is called {\em $b$-bad} for $S$ if 
      \[
       \left| f^{-1}(y) \cap S \right| \; > \; bD|S|/M,
       \]
      (thus the size exceeds $b$ times the average size for an element in $[M]$).
    \item 
      $x \in S$ is called {\em $\varepsilon$-poor} in $S$ if for more than
      a $2\varepsilon$-fraction of $i \in [D]$ the value $f(x,i)$ is $(1/\varepsilon)$-bad for~$S$.
  \end{itemize}
\end{definition}

\begin{lemma}\label{lem:extractorsAreGood}
  For every $(k,\varepsilon)$-extractor and set $S$, 
  less than $2^k$ elements in $S$ are $\varepsilon$-poor in~$S$.
\end{lemma}

\begin{proof}
  Let $f\colon [N] \times [D] \rightarrow [M]$ be a $(k,\varepsilon)$-extractor and $S \subseteq [N]$. 
  Choose $B \subseteq [M]$ to be the set of $(1/\varepsilon)$-bad elements for~$S$,
  (note that $|B| \le \varepsilon M$),
  and $A \subseteq [N]$ the set of $\varepsilon$-poor elements $x$ in $S$. 
  $\Pr \left\{ f(X_A, X_D) \in B \right\} \ge 2\varepsilon$ by choice of~$A$.
  Thus, $|A| < 2^k$ by definition of extractor graphs. 
  In other words, the number of $\varepsilon$-poor elements in $S$ is less than~$2^k$.
\end{proof}

\begin{proof}[Proof of  Lemma~\ref{lem:extractorSmallComplexity}.]
  Let $S = \left\{u \colon \C(u|\tilde{z}) \le \C(x|\tilde{z})\right\}$. Obviously, $x \in S$ and 
  $S$ can~be enumerated given $\tilde{z}$ and $\C(x|\tilde{z})$. 
  In a similar way, the $\varepsilon$-poor elements in $S$ can be enumerated.

  We show that {\em $x$ is not $\varepsilon$-poor in~$S$.} Assume $x$ were $\varepsilon$-poor.
  By  Lemma~\ref{lem:extractorsAreGood} the logarithm of the number of poor
  elements is at most $\C(x|\tilde{z})-c$, 
  and we can describe $x$ by its index in the list of $\varepsilon$-poor elements. Hence:
  \[
  \C(x|\tilde{z}, \C(x|\tilde{z})) \le \C(x|\tilde{z}) - c + O(1).
  \]
  The term in the left-hand is at least $\C(x|\tilde{z}) - O(1)$ by 
  the relativized version of the inequality $\C(w) \le \C(w | \C(w))+O(1)$. 
  We showed that $\C(x|\tilde{z}) - O(1) \le \C(x|\tilde{z}) - c$ 
  and this is a contradiction for large $c$. Hence, $x$ is not $\varepsilon$-poor.

  \smallskip
  For at least a $(1-2\varepsilon)$ fraction of $i$, 
  the value $f(x,i)$ is not $(1/\varepsilon)$-bad in $S$. We show that for such~$i$, the 
  complexity bound of  Lemma~\ref{lem:extractorSmallComplexity} is satisfied.
  Given $\tilde{z}$, $\C(x|\tilde{z})$ and $f(x,i)$ we can enumerate~$f^{-1}(f(x,i)) \cap S$,
  and describe $x$ by its index in this enumeration. Let us bound the size of this set.
  By definition of $(1/\varepsilon)$-bad the size of this set is at most $\frac{D|S|}{\varepsilon M}$, 
  and $|S| \le 2^{\C(x|\tilde{z})+1}$ by choice of~$S$. To the index of $x$, we need to prepend a description of
  $\C(x|\tilde{z})$ which requires $O(\log |x|)$ bits. In total, the length of the description of $x$
  satisfies the bound of  Lemma~\ref{lem:extractorSmallComplexity}.
\end{proof}

We apply  Lemma~\ref{lem:extractorSmallComplexity} to 
polynomial time computable extractor functions which are given in \dots.  
For a function $f$ whose values are strings of at least length $k$, let $f^{[k]}$ be the function 
obtained by taking the $k$-bit prefixes of values of $f$.
It is not hard to adapt the definition of extractors for functions from strings to strings.

\begin{theorem}\label{th:existMonotoneExtractors}
  There exists a polynomial time computable family of functions
  $f_{n,\varepsilon}: \{0,1\}^n \times [D] \rightarrow \{0,1\}^n$ with $D \le 2^{O(\log^2 (n/\varepsilon))}$
  such that $f_{n,\varepsilon}^{[k]}$ is a $(k,\varepsilon)$-extractor for all $k \le n$.
\end{theorem}


\begin{proof}[Proof of  Proposition~\ref{prop:extractorSmallComplexity}.]
  It suffices to prove the corollary for all $\varepsilon$ such that $\C(\varepsilon) \le O(\log (1/\varepsilon))$.
  Let $n = |x|$.

  \medskip
  \noindent
  Let $\alg F$ be the algorithm that on input $\varepsilon>0, k, x$ 
  evaluates the function $f_{n,\varepsilon}^{[k]}$ from Theorem~\ref{th:existMonotoneExtractors} 
  in the value $(x,i)$ for a random $i \in [D]$. 
   
  \medskip
  Let $f = f_{n,\varepsilon}$ and note that $\C(f) \le O(\log \tfrac{n}{\varepsilon})$.
  First assume that $k \le \C(x|\tilde{z})-c$. The conditions of Lemma~\ref{lem:extractorSmallComplexity} 
  are satisfied for the function $f^{[k]}$ with $M = 2^k$. Hence,
  \[
  \C\left(x|f^{[k]}(x,i), \tilde{z} \right) \le \C(x|\tilde{z}) - k + O(\log^2 \tfrac{n}{\varepsilon}).
  \]
  Because the complexities of $f^{[k]}$ and $\varepsilon$ are logarithmic, we obtain the same result if 
  we replace $\tilde{z}$ by $z$. In this case the corollary is proven.

  \medskip
  Let $\ell = \C(x|\tilde{z})-c$. Assume $k>\ell$. 
  In this case, we apply  Lemma~\ref{lem:extractorSmallComplexity}  to $f^{[\ell]}$
  and obtain 
  \[
  \C\left(x|f^{[\ell]}(x,i), z \right) \le c + O(\log^2 \tfrac{n}{\varepsilon}).
  \]
  This relation remains true if we replace $f^{[\ell]}$ by $f^{[k]}$ because the former is a prefix of the latter, 
  and the complexity can at most increase by a term $O(\log k)$.
   Proposition~\ref{prop:extractorSmallComplexity}  and hence  Theorem~\ref{th:fastCompression} are proven.
%
%
\end{proof}

\medskip
Proposition~\ref{prop:extractorSmallComplexity}  has another consequence that we 
use in the next section: if $x$ has large complexity, then with high probability $f(x,i)$ is random.
\begin{corollary}\label{cor:incompressibleDescription}
  If $\C(x|\tilde{z}) \ge k$, then
  with probability $1-\varepsilon$,
  $\varalg{F}_{\varepsilon,k,x}$ in  Proposition~\ref{prop:extractorSmallComplexity} satisfies:
  \[
  \C(\varalg{F}_{\varepsilon,k,x}|z) \;\,\ge\;\, k - O(\log^2 \tfrac{|x|}{\varepsilon}).
  \]
\end{corollary}

\begin{proof}
  If adding a string in the condition decreases the complexity of~$x$ by $k$, 
  then the string should have complexity at least $k-O(1)$ by symmetry of information.
\end{proof}

\section{Slepian-Wolf coding}

For a tuple $(x_1, \dots, x_\ell)$ and a set $S \subseteq [\ell]$, 
let $x_S$ be the tuple containing the strings $x_i$ with $i \in S$. 

\begin{theorem}\label{th:SlepianWolf}
  There exists a deterministic decompression algorithm $\alg D$ and 
  a probabilistic polynomial time compression algorithm $\alg E$ 
  that maps $\varepsilon > 0, \ell, k$ and a string $y$ 
  to a string of length $k + O(\ell \log^2 (|y|/\varepsilon))$ such that if 
  \begin{equation}\label{eq:conditionsSlepianWolf}
    \C(x_S| x_{[\ell]\setminus S}, z) \le \sum_{i \in S} k_i \quad\quad \text{for all }S \subseteq [\ell]
  \end{equation}
  for some $z$, $(k_1, \dots, k_\ell)$ and a tuple $x = (x_1,\dots, x_\ell)$ of strings of equal length,  then 
  with probability $1-\varepsilon$: $\alg D (E_1, \dots, E_\ell,z) = x$, 
  where $E_i$ is the output of $\alg E$ on input $\varepsilon,\ell, k_i, x_i$.
\end{theorem}

\noindent
The proof proceeds by induction on $\ell$ and in the induction
step we use the lemma below.
Let $\alg F$ be the algorithm of  Proposition~\ref{prop:extractorSmallComplexity} and 
$F_i$ the result of the algorithm on input $\varepsilon,k_i,x_i$.

\begin{lemma}\label{lem:SlepianWolfInductionStep}
  Let $b = c\ell\log^2 \tfrac{|x_i|+2}{\epsilon}$ for some large $c$.
  If \eqref{eq:conditionsSlepianWolf} is satisfied for all $S \subseteq [\ell]$, 
  then with probability $1-\epsilon$, for all $S \subseteq [\ell-1]$, 
  \eqref{eq:conditionsSlepianWolf} is satisfied 
  for the tuples $(x_1,\dots, x_{\ell-1})$ and $(k_1+b, \dots, k_{\ell-1} + b)$ 
  with condition $(z,\varalg F_\ell)$ (in stead of $z$).
\end{lemma}

\begin{proof}
  We prove the lemma for the special case that its assumption \eqref{eq:conditionsSlepianWolf} 
  has no $z$ in the condition.
  The proof of the conditional version follows the unconditional one.
  Let $S \subseteq [\ell-1]$ and $T = [\ell-1] \setminus S$. Let $n$ be the length of the strings. 
  Let $b$ be large enough so that it exceeds the $O(\cdot)$-terms for $\varepsilon = \epsilon2^{-\ell}$  
  in Proposition~\ref{prop:extractorSmallComplexity} and Corollary~\ref{cor:incompressibleDescription}. 
  We show that 
  \[
  \C(x_S|x_T, {F}_\ell) \le \sum_{i \in S} k_i + O(b)
  \]
  with probability~$1-\epsilon 2^{-\ell}$.  The lemma follows by applying this for all $T \subseteq [\ell]$ 
  and after a redefinition of~$b$.

  First suppose that $k_\ell \ge \C(x_\ell|x_T)$, 
  then $\C(x_\ell|F_\ell,x_T) \le b$ by Proposition~\ref{prop:extractorSmallComplexity}.
  Changing $x_\ell$ to ${F}_\ell$ in the condition of  \eqref{eq:conditionsSlepianWolf}, 
  changes the complexity at most by $O(b)$. In this case the equation above is proven.

  Now suppose that $k_\ell < \C(x_\ell|x_T)$. We apply symmetry of information.
  \[
    \C(x_S|x_T, {F}_\ell) = \C(x_S, {F}_\ell | x_T) - \C({F}_\ell|x_T) + O(\log (\ell n)).
  \]
  We can compute the value of ${F}_\ell$ from $x_\ell$ using $\log D$
  bits of information, hence the first term in the right-hand side is bounded by 
  $\C(x_S, x_\ell|x_T) \le k_S + k_\ell$ up to an $O(b)$-term. 
  By our assumption and  Corollary~\ref{cor:incompressibleDescription}, we have that the 
  second term is less than $-k_\ell$.
  We conclude that the left-hand side is bounded by $ (k_S + k_\ell) - k_\ell + O(b)$.
\end{proof}

\begin{proof}[Proof of  Theorem~\ref{th:SlepianWolf}.]
  Let $b$ be the value of Lemma~\ref{lem:SlepianWolfInductionStep}, and let $\tilde{\alg{D}}$ and $\tilde{\alg{E}}$ 
  be the algorithms given by  Theorem~\ref{th:fastCompression}.

  On input an $\ell$-tuple of strings $(E_1,\dots,E_\ell)$, and $z$, {\em algorithm $\alg{D}$} 
  outputs the $\ell$-tuple containing the results of running $\tilde{\alg D}$ 
  on the arguments $E_i$ and $({E}_{[\ell] \setminus i},z)$ for~$i = 1,\dots, \ell$. 
  On input $\varepsilon,\ell,k,x$,  {\em algorithm $\alg E$} runs $\tilde{\alg E}$ 
  on inputs $\varepsilon, k + (\ell - 1)b$ and~$x$. 

  We show that this works. By $\ell-1$ applications of  Lemma~\ref{lem:SlepianWolfInductionStep} we have 
  \[
  \Pr \left\{ \C(x_i | {F}_{[\ell] \setminus i}) \le k_i + (\ell-1)b \right\} \;\; \ge \;\; 1-(\ell-1)\varepsilon
  \]
  This implies that for a fixed $i$, on input $E_i$ and ${F}_{[\ell] \setminus i}$,  
  algorithm $\tilde{\alg{D}}$ fails to return $x_i$ with probability at most~$\ell\varepsilon$. 
  Hence, the output of $\alg D$ equals $(x_1,\dots, x_\ell)$ 
  with probability $1-\ell^2\varepsilon$ and the result follows after rescaling~$\varepsilon$. 
  Observe that if $\ell > n$, the theorem is trivially true (because $\alg E(x) = x$ works), thus this rescaling 
  does not change the $O(\ell \log^2 (n/\varepsilon))$-term.
\end{proof}

\bibliography{eigen,kolmogorov}
\end{document}